\documentclass[a4wide]{article}

\usepackage{macros}
\usepackage{datetime}

\allowdisplaybreaks[4]

\usepackage{parskip}
\usepackage{hyperref}

\usepackage[a4paper]{geometry}

\newcommand{\eod}{\hfill\ensuremath{\diamond}}

\newtheoremstyle{myplain}
{5pt}
{5pt}
{}
{}
{\bf}
{.}
{\newline}
{}

\theoremstyle{myplain}
\newtheorem{theorem}{Theorem}[section]
\newtheorem{lemma}[theorem]{Lemma}

\newtheorem{corollary}[theorem]{Corollary}

\newtheorem{definition}[theorem]{Definition}

\newtheorem{example}[theorem]{Example}

\begin{document}
\title{An Extension of Parikh's Theorem beyond Idempotence\thanks{This work was partially funded by DFG project ``Polynomielle Systeme {\"u}ber Semiringen: Grundlagen, Algorithmen, Anwendungen''}}

\author{Michael Luttenberger and Maxmilian Schlund\thanks{Institut f{\"u}r Informatik, Technische Universit{\"a}t M{\"u}nchen}\\ \texttt{\{luttenbe,schlund\}@model.in.tum.de}}

\maketitle

\begin{abstract}
The {\em commutative ambiguity} $\camb_{G,X}$ of a context-free grammar $G$ with start symbol $X$ assigns to each Parikh vector $\vv$ the number of distinct leftmost derivations yielding a word with Parikh vector $\vv$. Based on the results on the generalization of Newton's method to $\omega$-continuous semirings \cite{EKL07:stacs,EKL07:dlt,DBLP:journals/jacm/EsparzaKL10}, we show how to approximate $\camb_{G,X}$ by means
of rational formal power series, and give a lower bound on the convergence speed of these approximations. From the latter result we deduce that $\camb_{G,X}$ itself is 
rational modulo the generalized idempotence identity $k=k+1$ (for $k$ some positive integer), and, subsequently, that it can be represented as a weighted sum of linear sets. This extends Parikh's well-known result that the commutative image of context-free languages is semilinear ($k=1$).

Based on the well-known relationship between context-free grammars and algebraic systems over semirings \cite{ChomskySch1963,DBLP:books/daglib/0067812,DBLP:journals/tcs/BerstelR82,Kui,DBLP:journals/mst/Bozapalidis99},
our results extend the work by Green et al.\ \cite{DBLP:conf/pods/GreenKT07} on the computation of the provenance of Datalog queries over commutative $\omega$-continuous semirings.
\end{abstract}

\section{Introduction}

\newcommand{\coloneq}{\ensuremath{\mathrel{\mathop:}-}}

\paragraph*{Motivation}
Recently, Green et al.\ showed in \cite{DBLP:conf/pods/GreenKT07} that several questions regarding the provenance of an answer to a Datalog query 
\footnote{See e.g.\ \cite{DBLP:journals/tkde/CeriGT89} for more details on Datalog.}
reduce to computing the least solution of an algebraic system over a $\omega$-continuous commutative semiring. 
To illustrate the main idea, consider the following Datalog program that computes the transitive closure of a finite directed graph ${\cal G}=(V,E)$:
\begin{equation*}
\begin{array}{lcl}
  \text{trans}(X,Y) & \coloneq & \text{edge}(X,Y).\\
  \text{trans}(X,Y) & \coloneq & \text{trans}(X,Z), \text{trans}(Z,Y).\\
\end{array}
\end{equation*}
Here, $X,Y,Z$ are variables ranging over the nodes $V$ of the graph, the interpretation of the (extensional) predicate $\text{edge}(X,Y)$ is given by the edge relation $E$ of ${\cal G}$,
while the interpretation of the (intensional) predicate $\text{trans}(X,Y)$ is implicitly given by the least Herbrand model, i.e.\ the transitive closure of ${\cal G}$.
In order to deduce which edges of ${\cal G}$
give rise to a positive answer to the query $?-\text{trans}(u,v).$, in \cite{DBLP:conf/pods/GreenKT07} the authors assign to each positive literal a unique identifier
 -- for instance, let $\al = \{ e_{u,v} \mid (u,v) \in E\}$ and $\vars = \{ X_{u,v} \mid u,v\in V\}$ -- 
and then expands the above query into an abstract algebraic system in the formal parameters $\al$ and the variables $\vars$:
\begin{equation*}
 X_{u,w} = \left\{ \begin{array}{ll} e_{u,w} + \sum_{v\in V} X_{u,v} X_{v,w} & \text{ if } (u,w) \in E\\
                                        \sum_{v\in V} X_{u,v} X_{v,w} & \text{ otherwise } \end{array}\right.
\end{equation*}
In order to give a meaning to this system, the right-hand side is interpreted over some semiring $\cg{S,+,\cdot,0,1}$, short $S$,
i.e.\ the abstract addition and multiplication are interpreted as the addition and multiplication in $S$, and each formal parameter $a\in\al$ is interpreted as an element $h(a)\in S$ by
means of a valuation $h\colon \al\to S$.
As is well-known \cite{Kui}, each algebraic system has a least solution if $S$ is $\omega$-continuous (see Section \ref{sec:pre}).

We demonstrate the connection between the Datalog program and the algebraic system by means of two examples. First,
the transitive closure itself is essentially the least solution over the Boolean semiring $\cg{\{0,1\},\vee,\wedge,0,1}$ under the valuation $h(e_{u,w})=1$ for all $e_{u,w}\in\al$, i.e.\
the least solution assigns $1$ to $X_{u,w}$ if and only if $(u,w)$ is in the transitive closure.
For a somewhat more interesting example,
assume we want to analyze {\em why} an edge $(u,w)$ is included in the transitive closure.
To this end, it suffices to represent a path by the set of its edges, and a set of paths
by the set of corresponding sets of edges. This leads naturally to the semiring $\cg{2^{2^{\al}},\cup, \Cup , \emptyset, \{\emptyset\}}$: a semiring element is a set of subsets of edge identifiers, two semiring elements $s_1,s_2$ are added by taken their union $s_1\cup s_2$, while the (commutative) multiplication is defined by $s_1 \Cup s_2 = \{ a_1 \cup a_2 \mid a_1 \in s_1, a_2 \in s_2\}$.  Again, we obtain the answer to our question by computing the least solution of above system over this semiring under the valuation $h(e_{u,w})=\{\{e_{u,w}\}\}$.
For further examples, we refer the reader to \cite{DBLP:conf/pods/GreenKT07}.

Note that in both examples, multiplication is commutative, and addition is idempotent. 
Naturally, the question arises over which commutative $\omega$-continuous semirings
we can compute or, at least, approximate the least solution of an algebraic system. Of particular interest is the semiring of {\em formal power series} whose carrier is the set $\Next\aab{\N^\al}$ of functions from Parikh vectors $\N^\al$ to the extended natural numbers $\Next=\N\cup\{\infty\}$, as it is free in the following sense:
every valuation $h\colon \al \to S$ into a concrete commutative $\omega$-continuous semiring induces a unique $\omega$-continuous homomorphism $H\colon \Next\aab{\al^\ast}\to S$
 which maps the least solution over $\Next\aab{\N^\al}$ to the least solution over $S$ (we do not distinguish between $h$ and $H$ in the following). See e.g.\ \cite{DBLP:journals/mst/Bozapalidis99,DBLP:conf/pods/GreenKT07}. 

In general, a finite, explicit representation of the least solution $(\mf{s}_X\mid X\in\vars)$ over $\Next\aab{\N^\al}$ is not possible (see also Example \ref{ex1}).
In \cite{DBLP:conf/pods/GreenKT07} the authors therefore present two algorithms {\em All-Trees} and {\em Monomial-Coefficient} for computing finitely representable information on this solution: {\em All-Trees} decides whether $\mf{s}_X\colon \N^\al \to \Next$ 
has only finite support and takes only finite values on its support, and can be used to evaluate Datalog over finite distributive lattices, a special case of commutative $\omega$-continuous semirings; {\em Monomial-Coefficient} computes the value of $\mf{s}_X$ for some Parikh vector $\vv\in\N^\al$.
Both algorithms are based on the close relationship between algebraic systems and context-free grammars \cite{ChomskySch1963,DBLP:books/daglib/0067812,Kui,cfgHFL,DBLP:journals/jcss/Thatcher67,DBLP:journals/tcs/BerstelR82,DBLP:journals/mst/Bozapalidis99,EKL07:stacs,EKL07:dlt,EKL08:dlt},
and work by enumerating the derivation trees of the grammar associated with the algebraic system utilizing the pumping lemma for context-free languages in order to ensure termination.
The associated context-free grammar $G=(\vars,\al,P)$ with nonterminals $\vars$, alphabet $\al$, and productions $P$ is obtained from the algebraic system by reinterpreting the right-hand sides of the algebraic system as rewriting rules for the variables. For instance, the algebraic system for computing the transitive closure translates to the grammar $G$ defined by the rules
\begin{equation*}
X_{u,w} \to X_{u,v} X_{v,w}\ \text{ for all $u,v,w\in V$, and }\ X_{u,w} \to e_{u,w}\ \text{ for all $(u,w)\in E$.}
\end{equation*}
W.r.t.\ commutative $\omega$-continuous semirings, 
the grammar $G$ and the algebraic system are then connected by means of the commutative ambiguity $\camb_{G,X}\colon \N^\al \to \Next$ which assigns to each Parikh vector $\vv\in\N^\al$ the number of leftmost derivations w.r.t.\ $G$ with start symbol $X$ leading to a word with Parikh vector $\vv$: we have that $\mf{s}_X = \camb_{G,X}$ for all $X\in\vars$, or short $\mf{s} = \camb_G$. See e.g.\ \cite{ChomskySch1963,DBLP:journals/mst/Bozapalidis99,EKL07:stacs}.

\paragraph*{Contribution and related work} 
In this article, we study how to construct from a given context-free grammar $G$ a sequence $G^{\dle{0}}, G^{\dle{1}}, \ldots$ of {\em nonexpansive} context-free grammars $G^{\dle{i}}$ that underapproximate the ambiguity of $G$ ($\amb_{G^{\dle{i}},X}(w) \le \amb_{G,X}(w)$ for all $w\in\al^\ast$, Lemma \ref{lem:unf}), and, thus, also the commutative ambiguity.\footnote{A context-free grammar is nonexpansive  if every variable $X$ derives only sentential forms containing $X$ at most once \cite{GinsburgSpanier68:deriv-bounded}.} As $G^{\dle{i}}$ is nonexpansive, it is straightforward to show that $\camb_{G^{\dle{i}},X}$ is rational in $\Next\aab{\N^\al}$,
and a rational expression representing $\camb_{G^{\dle{i}},X}$ can easily be computed from $G^{\dle{k}}$ (Theorem \ref{thm:nonexp}). We then give a lower bound on the speed at which $\camb_{G^{\dle{i}},X}$ converges to $\camb_{G}$: letting $n$ be the number of variables of $G$,
we show that 
for every positive integer $k$ and every $\vv\in\N^\al$ we have that, if $\camb_{G{^{\dle{n+k}}},X}(\vv) \neq \camb_{G,X}(\vv)$, then at least $2^{2^k}\le \camb_{G^{\dle{n+k}},X}(\vv)$  (Theorem \ref{thm:parikh-ext}).

An immediate consequence of these results is an algorithm for evaluating Datalog queries over ``collapsed'' commutative semirings: call a $\omega$-continuous semiring $S$ collapsed at some positive integer $k$ if in $S$ the identity $k=k+1$ holds;\footnote{Where $k$ denotes the term $1+\ldots+1$ consisting of the corresponding number of $1$s. For instance, any $\omega$-continuous idempotent semiring is ``collapsed'' at $1$. See also \cite{DBLP:journals/iandc/BloomE09} for a much more general discussion of these semirings.} given a valuation $h\colon \al \to S$ into a commutative $\omega$-continuous semiring collapsed at $k$, the least solution can be obtained by evaluating the corresponding rational expressions for $\camb_{G^{\dle{n + \log\log k}}}$ under the homomorphism induced by $h$.
 
In particular, this yields an algorithm for evaluating Datalog queries over the tropical semiring $\cg{\Next,\min,+,0,\infty}$; this answers an open question of \cite{DBLP:conf/pods/GreenKT07}. We remark that in \cite{EKL08:dlt} more efficient algorithms for the classes of star-distributive semirings, subsuming the tropical semiring, and of one-bounded semirings, subsuming finite distributive lattices, are presented.

Finally, we show that $\camb_{G,X}$ can be represented modulo $k=k+1$ as a finite sum $\gamma_1 \vecn{1}_{C_1} + \ldots + \gamma_r \vecn{1}_{C_r}$ of weighted characteristic functions $\vecn{1}_C$ of linear sets $C\subseteq\N^\al$ with weights $\gamma_i\in\{0,1,\ldots,k\}$ (Theorem \ref{thm:semilin}).\footnote{$C\subseteq\N^\al$ is linear if $C=\{ \vv_0 + \sum_{i=1}^s \lambda_i \vv_i \mid \lambda_1,\ldots,\lambda_s\in\N\}$ for vectors $\vv_0,\ldots,\vv_s\in \N^\al$.} This completes the extension of Parikh's well-known theorem that the commutative image of a context-free grammar is a semilinear set ($k=1$).

These results continue the study of Newton's method over $\omega$-continuous semirings presented in \cite{EKL07:stacs,EKL07:dlt,DBLP:journals/jacm/EsparzaKL10}. There it was shown that Newton's method, as known from calculus, also applies to the setting of algebraic systems over $\omega$-continuous semirings,
and converges always to the least solution at least as fast as (and many times much faster than) the standard fixed-point iteration.
Although it is shown in \cite{EKL07:dlt,DBLP:journals/jacm/EsparzaKL10} that Newton's method is well-defined on any $\omega$-continuous semiring,
the definition does not yield an effective way of applying Newton's method as it requires the user to supply at each iteration a semiring element which
represents a certain difference. 
Only for special cases it is stated how to compute those differences,
but a general construction is missing in these articles.

The grammars $G^{\dle{k}}$ defined in Definition \ref{def:unf} address this shortcoming. Their construction is based on the notion of ``tree dimension'' 
introduced in \cite{EKL07:stacs} to characterize the structure of terms evaluated by Newton's method, where it was shown that the $k$-th Newton approximation of the least solution of an algebraic system
corresponds exactly to the derivation trees of dimension at most $k$ generated by the context-free grammar associated with the system.
This allows us to explicitly define a grammar, resp.\ equation system, which captures exactly the update computed by Newton's method
within a single step. That is, we may define the difference of two consecutive Newton approximations over {\em any} $\omega$-continuous semiring by constructing a grammar which generates exactly the derivation trees of $G$ of dimension exactly $k$.
By taking the sum of all these updates, we obtain the grammar, $G^{\dle{k}}$ which generates exactly the derivation trees of $G$ of dimension at most $k$.
Hence, if the least solution of (the equation system associated with) $G^{\dle{k-1}}$ is known, we only need to solve the equation system corresponding to
the derivation trees of dimension exactly $k$.
We remark that this construction does not require multiplication to be commutative; it is merely a partition of the regular tree language of derivation trees of $G$.

If multiplication is commutative, $\camb_{G^{\dle{k}}}$ represents the $k$-th Newton approximation over any commutative $\omega$-continuous semiring. Similarly, the bound on the speed at which $\camb_{G^{\dle{k}}}$ converges to $\camb_G$ given in Theorem \ref{thm:parikh-ext} generalizes the result of \cite{EKL07:stacs} on the convergence of Newton's method over idempotent commutative $\omega$-continuous semirings.

If multiplication is not commutative, we may not represent the least solution of $G^{\dle{k}}$ as regular expressions, but only as regular tree expressions with the particular property
that tree substitution only occurs at a unique leaf. It might be worthwhile to study
if there are interesting (distributive) abstract interpretations whose widening operator can take advantage of this representation.

\paragraph*{Structure of the paper}
In Section \ref{sec:pre} we recall the most fundamental definitions, in particular the definition of the dimension of a tree.
We then show in Section \ref{sec:unf} how to unfold a given context-free grammar $G$ into a new context-free grammar $G^{\dle{k}}$
that generates exactly those derivation trees of $G$ that are of dimension at most $k$ and, thus, represents exactly the $k$-th Newton approximation.
We show that the commutative ambiguity of each grammar $G^{\dle{k}}$ is rational over $\Next\aab{\N^\al}$.
In Section \ref{sec:conv} we give a lower bound on the speed at which the ambiguity of $G^{\dle{k}}$ converges to that of $G$.
We use this result in Section~\ref{sec:semilinear} to obtain from a rational expression for $\camb_{G^{\dle{k}}}$ a semilinear representation of $\camb_{G}$ modulo the generalized idempotence assumption of $k=k+1$, thereby completing the extension of Parikh's theorem from $k=1$ to arbitrary $k$. 

All proofs can be found in the appendix.

\section{Preliminaries}\label{sec:pre}
The power set of a set $M$ is denoted by $2^M$. For $k\in\N$, set $[k] := \{1,2,\ldots,k\}$ with $[0]=\emptyset$.
The 
natural numbers extended by a greatest element $\infty$, and the natural numbers ``collapsed'' at a given positive integer $k$ are denoted by 
$\Next$, and $\N_k = \{0,1,\ldots,k\}$, respectively. For $a\in\Next$ set $a+\infty = \infty$,
$0\cdot \infty = 0$ and $a\cdot \infty = \infty$ if $a\neq 0$. Addition and multiplication are defined on $\N_k$ by 
identifying $k$ with $\infty$.

The set of words over the (finite) alphabet $\al$ is denoted by $\al^\ast$ with $\ew=()$ the empty word.
The length of a word $w\in\al^\ast$ is denoted by $\abs{w}$. The {\em Parikh map} is $\pa\colon \al^\ast \to \N^\al \colon w \mapsto (\pa_a(w) \mid a\in\al)$ 
where $\pa_a(w)$ denotes the number of occurrences of $a$ in $w$.

Let $\sig$ be finite ranked set (signature) where $\sig_r$ denotes the subset of $\sig$ consisting of exactly those symbols having arity $r$.
Then $\trs$ denotes the set of $\sig$-terms where we use Polish notation so that $\trs\subseteq \sig^\ast$.
When $t\in\trs$, we denote by $t=\sigma t_1\ldots t_r$ that $\sigma\in\sig_r$ and $t_1,\ldots,t_r\in\trs$ are the uniquely determined subterms;
for inductive definitions, we set $t=\sigma t_1\ldots t_r = \sigma$ if $r=0$.
$\trs$ is canonically identified with the set of finite, $\sig$-labeled, rooted trees: the rooted tree underlying $t=\sigma t_1\ldots t_r$ has as nodes the set $V_t = \{ \ew \} \cup \{ i \pi \mid i\in [r],\ \pi\in V_{t_i} \}$ with $\ew$ the root,
and the edges $E_t := \{ (\pi,\pi i) \mid \pi i \in V_t\}$ pointing away from the root. 
The label $\lbl_t(\cdot)$ of a node in $V_t$ is then defined inductively by $\lbl_{t}(\ew) = \sigma$ 
and $\lbl_{t}(i\pi) = \lbl_{t_i}(\pi)$ for $t=\sigma t_1\ldots t_r$. 
The height $\hgt(t)$ of a tree $t=\sigma t_1\ldots t_r$ is defined to be $0$ if $r=0$, and otherwise by $\hgt(t)=\max_{i\in[r]} \hgt(t_i)$.
Analogously, define the subtree $t|_{\pi}$ of $t$ rooted at $\pi$, and the tree $t[t'/\pi]$ obtained by substituting the tree $t'$ for $t|_{\pi}$ inside of $t$.

\begin{definition}
The {\em dimension} $\dim(t)$ of $t=\sigma t_1\ldots t_r\in\trs$ is defined
to be $\dim(t) = 0$ if $r=0$; otherwise let $d = \max_{i\in[r]} \dim(t_i)$, and set $\dim(t)=d$ if there is a unique child $i\in[r]$ of dimension $d$, 
else set $\dim(t) = d+1$.
\eod
\end{definition}

From the definition it easily follows that $\dim(t)$ is the height of the greatest perfect binary tree that can be obtained from the rooted tree $(V_t,E_t)$ via edge contractions.
Thus, $\dim(t)$ is bounded from above by $\hgt(t)$. 

\begin{example}
Assume $\sig=\{a,b\}$ with $a\in\sig_2$ and $b\in\sig_0$. Then $aabbaabbb\in\trs$ is identified with the tree
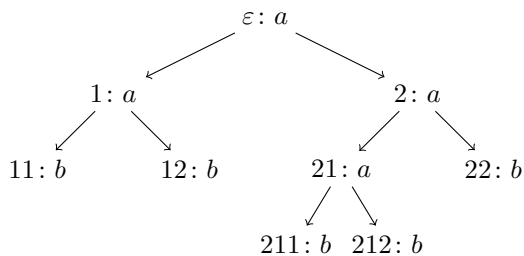
\begin{center}
\begin{tikzpicture}
\node (e) at (0,0) {$\ew \colon a$};
\node (0) at (-2,-1) {$1\colon a$};
\node (1) at (2,-1) {$2\colon a$};
\node (00) at (-3,-2) {$11\colon b$};
\node (01) at (-1,-2) {$12\colon b$};
\node (10) at (1,-2) {$21\colon a$};
\node (11) at (3,-2) {$22\colon b$};
\node (100) at (0.4,-3) {$211\colon b$};
\node (101) at (1.6,-3) {$212\colon b$};

\draw[->] (e) -> (0);
\draw[->] (e) -> (1);
\draw[->] (0) -> (00);
\draw[->] (0) -> (01);
\draw[->] (1) -> (10);
\draw[->] (1) -> (11);
\draw[->] (10) -> (100);
\draw[->] (10) -> (101);
\end{tikzpicture}
\end{center}
For instance, the node $212$ is labeled by $b$. Computing the dimension bottom-up, we obtain $\dim(t|_{21})=1$, $\dim(t|_2)=1$, $\dim(t|_1)=1$, and $\dim(t)=2$.
\end{example}

The tree dimension $\dim(t)$ is also known as Horton-Strahler number \cite{Horton45,Strahler52}, or the register number \cite{DBLP:journals/cacm/Ershov58,DBLP:conf/focs/FlajoletFV79,DBLP:journals/ipl/DevroyeK95}, and is closely related to the pathwidth \cite{DBLP:journals/jct/RobertsonS83} $\pw(T)$ of the tree $T=(V_t,E_t)$ underlying $t$: it can be shown that $\pw(T) - 1 \le \dim(t) \le 2\pw(T)+1$. 

\paragraph*{Semirings} 
We recall the basic results on semirings (see e.g.\ to \cite{Kui,KuiHWA}).
A {\em semiring} $\cg{S,+,\cdot,0,1}$ consists of a commutative additive monoid $\cg{S,+,0}$ and a multiplicative monoid $\cg{S,\cdot,1}$ where multiplication distributes over addition from both left and right, and multiplication by $0$ always evaluates to $0$. We simply write $S$ for $\cg{S,+,\cdot,0,1}$ if the signature is clear from the context.
$S$ is {\em commutative} if its multiplication is commutative.
$S$ is {\em naturally ordered} if the relation $a \no b$ defined by $a\no b :\Leftrightarrow \exists d\in S\colon a+d = b$ is a partial order on $S$; then $0$ is the least element. 

A partial order $\cg{P,\le}$ is {\em $\omega$-continuous} if for every monotonically increasing sequence ({\em $\omega$-chain}) $(a_i)_{i\in\N}$, i.e.\ $a_i\le a_{i+1}$ for all $i\in\N$, the supremum $\sup_{i\in\N} a_i$ exists in $\cg{P,\le}$; a function $f\colon \cg{P,\le} \to \cg{P,\le}$ is called {\em $\omega$-continuous} if for every $\omega$-chain $(a_i)_{i\in\N}$ we have $f(\sup_{i\in\N} a_i) = \sup_{i\in\N} f(a_i)$. We say that $S$ is {\em $\omega$-continuous} if $\cg{S,\no}$ is $\omega$-continuous, and addition and multiplication are both $\omega$-continuous in every argument.
In any $\omega$-continuous semiring finite summation $\sum$ can be extended to countable sequences and families by means of $\sum_{i\in\N} a_i := \sup_{k\in\N} \sum_{i\in[k]} a_i$.
The {\em Kleene star} ${}^\ast\colon S \to S$ is defined by $a^\ast := \sum_{i\in\N} a^i$. 

If not stated otherwise, we always assume that $\Next$ carries the semiring structure $\cg{\Next,+,\cdot,0,1}$ with 
addition and multiplication as stated above so that $1^\ast = \infty$. 
For any $\omega$-continuous semiring $S$ there is exactly one $\omega$-continuous homomorphism $h$ from $\Next$ to $S$ as $h(0)=0$, $h(1)=1$, and $h(\infty)=h(1^\ast)=1^\ast$ have to hold; we therefore embed $\Next$ into $S$ by means of this unique homomorphism.

For a commutative semiring $\cg{S,+,\cdot,0,1}$, and a {\em finitely decomposable}\footnote{A monoid $\cg{M,\circ,e}$ is {\em finitely decomposable} if for every $m\in M$ there exists only finitely many pairs $(u,v)\in M^2$ 
that $u\circ v = m$. This ensures that the Cauchy product is also well-defined over semirings $S$ which are not $\omega$-continuous.} monoid $\cg{M,\circ,e}$  we recall the  definition of the {\em semiring $S\aab{M}$ of formal power series}. Its carrier is the set of total functions from $M$ to $S$. For $\mf{s}\in S\aab{M}$  denote by $(\mf{s},m)$ the value of $\mf{s}$ at $m\in M$. Then addition on $S$ is extended pointwise to $S\aab{M}$, while multiplication is defined by means of the generalized Cauchy product, i.e.:
\begin{equation*}
(\mf{s}+\mf{t},m) = (\mf{s},m) + (\mf{t},m)\ \text{ and }\ (\mf{s}\cdot\mf{t},m) = \sum_{u,v\in M\colon u\circ v = m} (\mf{s},u) \cdot (\mf{t},v).
\end{equation*}
That is, we treat $\mf{s}\in S\aab{M}$ as a (formal) power series $\sum_{m\in M} (\mf{s},m) m$ with $(\mf{s},m)$ the coefficient of the monomial $m$. 
If the support $\supp(\mf{s})=\{ m\in M\mid (\mf{s},m) \neq 0\}$ is finite, then $\mf{s}$ is called a (formal) polynomial. The subset of polynomials is denoted by $S\ab{M}$.
The semiring $S$ and the monoid $M$ are canonically embedded into $S\aab{M}$ by means of the  monomorphisms
$h_S\colon S \mapsto S\aab{M} \colon s \mapsto s e$  and $h_M \colon M \mapsto S\aab{M}\colon m \mapsto 1 m$, respectively.
W.r.t.\ these definitions $S\aab{M}$ and $S\ab{M}$ become semirings with neutral elements $\mf{0}=h_S(0)$ and $\mf{1}=h_S(1)=h_M(e)$; if $S$ is  $\omega$-continuous, then so is $S\aab{M}$, and the Kleene star is defined everywhere on $S\aab{M}$.
For instance, $S\aab{M}$ is $\omega$-continuous for $S$ either $\Next$ or $\N_k$, and $M$ either $\al^\ast$ or $\N^\al$;
but $\N\aab{\al^\ast}$ and $\N\aab{\N^\al}$ are not.
Note that $\Next\aab{\al^\ast}$ is free in the following sense:
let $\cg{S,+,\cdot,0_S,1_S}$ be some $\omega$-continuous semiring;
then every valuation $h\colon \al \to S$ extends uniquely to a $\omega$-continuous homomorphism $h\colon\Next\aab{\al^\ast}\to S$ defined by $h(\mf{s})=\sum_{w\in\al^\ast} (\mf{s},a) h(a)$. Similarly, $\Next\aab{\N^\al}$ is a representation of the free commutative $\omega$-continuous semiring generated by $\al$, and, thus, isomorphic to $\Next\aab{\al^\ast}$ modulo commutativity. 

Let $S$ be commutative and $\omega$-continuous so that the Kleene star is defined for every power series in $S\aab{M}$.
A power series $\mf{s}\in S\aab{M}$ is called {\em rational}, if it can be constructed from the elements of $S$ and $M$ by means of the rational operations addition, multiplication, and Kleene star, i.e.\ if either $\mf{r} \in S$, or $\mf{r} \in M$, or $\mf{r} = (\mf{r}_1+\mf{r}_2)$, or $\mf{r} = \mf{r}_1 \cdot \mf{r}_2$, or
$\mf{r} = \mf{r}_1^\ast$
for $\mf{r}_1,\mf{r}_2$ rational in $S\aab{M}$. A {\em rational expression (over $M$ with weights in $S$)} is any term constructed from elements of $S$ and $M$, and the rational operations. For every rational series $\mf{r}$ in $S\aab{M}$ there is a rational expression $\rho$ which evaluates to $\mf{r}$ over $S\aab{M}$.
By our assumption that $S$ is $\omega$-continuous, also every rational expression evaluates to a rational series $\mf{r}$ over $S\aab{M}$.
Note that $\omega$-continuous homomorphisms preserve rationality.

\paragraph*{Context-free grammars}
A context-free grammar $G=(\vars,\al,P)$ consists of variables $\vars$, an alphabet $\al$, and rules $P\subseteq \vars\times (\al\cup\vars)^\ast$. By $(G,X)$ we denote the grammar $G$ with start symbol $X\in\vars$.
For a rule $(X,\gamma)\in P$ we also write $X\to_G \gamma$ or simply $X\to \gamma$ if $G$ is apparent from the context.
$\Rightarrow_G$ denotes the binary relation on $(\al\cup\vars)^\ast$ induced by the rules $P$, i.e.,
if $X\to_G w$, then $\alpha X \beta \Rightarrow_G \alpha w \beta$ for all $\alpha,\beta\in (\al\cup\vars)^\ast$.
The (reflexive) transitive closure of $\Rightarrow_G$ is denoted by ($\Rightarrow_G^\ast$) $\Rightarrow_G^+$.
The language generated by $(G,X)$ is $L(G,X) = \{ w\in \al^\ast \mid X \Rightarrow_G^\ast w \}$.

Let $\sig_G$ denote the set $\{ \sigma_{X,\gamma} \mid X\to_G \gamma\}$
and define the arity of $\sigma_{X,\gamma}$ to be the number of variables occurring in $\gamma$.
Define the new context-free grammar $G_{\tr}$ with alphabet $\sig_G$ by setting 
$X\to_{G_{\tr}} \sigma_{X,\gamma} X_1 \ldots X_r \text{ for } \gamma = \gamma_0 X_1 \gamma_1 \ldots \gamma_{r-1} X_r \gamma_r$. 
Then $\tr_{G,X}:=L(G_{\tr},X)\subseteq \tr_{\sig_G}$ is called the set of $(G,X)$-trees (or simply $X$-trees if $G$ is apparent from the context)
and $\tr_{G,X}$ ``yields'' $L(G,X)$ in the sense of \cite{DBLP:journals/jcss/Thatcher67,DBLP:journals/tcs/BerstelR82,DBLP:journals/mst/Bozapalidis99,EKL07:stacs}:
The word represented by a tree $t\in\tr_{\sig_G}$ is called its yield $\Y(t)$ and is inductively defined by $Y(t) = u_0 Y(t_1) u_1 \ldots u_{r-1} Y(t_{r}) u_r$ for $t=\sigma_{X,\gamma}t_1 \ldots t_r$ and $\gamma= u_0 X_1 u_1 \ldots u_{r-1} X_{r} u_r$.
We then have $L(G,X) = \{ \Y(t) \mid t\in \tr_{G,X} \}$, and
\begin{equation*}
  \amb_{G,X}(w) = \abs{ \{ t\in\tr_{G,X} \mid \Y(t) = w \}}\ \text{ and }\ \camb_{G,X}(\vv) = \abs{ \{ t\in \tr_{G,X} \mid \pa(\Y(t)) = \vv \}}.
\end{equation*}
where $\amb_{G,X}\in\Next\aab{\al^\ast}$, $\camb_{G,X}\in\Next\aab{\N^\al}$ and $L(G,X)= \supp(\amb_{G,X}) \in \N_1\aab{\al^\ast}$.

The dimension of a derivation tree is closely related to the index of a derivation.

\begin{definition}[see e.g.\ \cite{GinsburgSpanier68:deriv-bounded}]
The {\em index} of a derivation is the maximal number of variables occurring in any sentential form of the derivation.\eod
\end{definition}

\begin{definition}
For $G$ a context-free grammar and $t\in\trsg$, let $\minidx(t)$ be the minimum index taken over all derivations associated with $t$.\eod
\end{definition}

\begin{lemma}[\cite{EKL07:dlt,EGKL10:parikh}]\label{lem:index}
Let $G$ be a context-free grammar and $r_{\max}$ the maximal arity of a symbol in $\sig_G$. Then:
$\dim(t) < \minidx(t) \le \dim(t) \cdot ( r_{\max} - 1) + 1$.
\eod
\end{lemma}

\begin{example}
Consider $G$ defined by the productions:
\begin{equation*}
X\to YaYaY\quad Y\to X\quad Y\to b.
\end{equation*}
Then $\sig_G = \{ \sigma_{X,XXX}, \sigma_{X,Y}, \sigma_{Y,a}\}$. The leftmost derivation 
\begin{equation*}
X\Rightarrow YaYaY \Rightarrow XaYaY \Rightarrow YaYaYaYaY \Rightarrow^+ babababab
\end{equation*}
has index $5$, and corresponds to the derivation tree 
\begin{equation*}
t = \sigma_{X,YaYaY}\ \sigma_{Y,X}\ \sigma_{X,YaYaY}\ \sigma_{Y,b}\ \sigma_{Y,b}\ \sigma_{Y,b}\ \sigma_{Y,b}\ \sigma_{Y,b}
\end{equation*}
depicted as
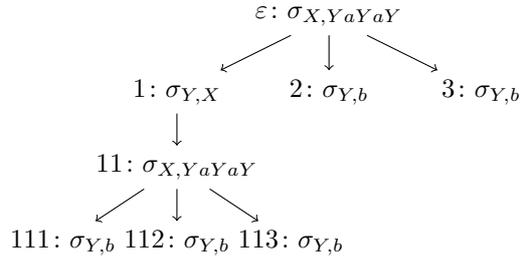
\begin{center}
\begin{tikzpicture}
\node (e) at (0,0) {$\ew\colon \sigma_{X,YaYaY}$};
\node (1) at (-2,-1) {$1\colon \sigma_{Y,X}$};
\node (2) at (0,-1) {$2\colon \sigma_{Y,b}$};
\node (3) at (2,-1) {$3\colon \sigma_{Y,b}$};
\node (11) at (-2,-2) {$11\colon \sigma_{X,YaYaY}$};
\node (111) at (-3.5,-3) {$111\colon \sigma_{Y,b}$};
\node (112) at (-2,-3) {$112\colon \sigma_{Y,b}$};
\node (113) at (-0.5,-3) {$113\colon \sigma_{Y,b}$};

\draw[->] (e) -> (1);
\draw[->] (e) -> (2);
\draw[->] (e) -> (3);
\draw[->] (1) -> (11);
\draw[->] (11) -> (111);
\draw[->] (11) -> (112);
\draw[->] (11) -> (113);
\end{tikzpicture}
\end{center}
This tree has dimension $1$. A derivation of minimal index first processes the subtree $t|_{2}$ and $t|_{3}$ leading to an index of $3$.
\end{example}

\section{Unfolding}\label{sec:unf}
In this section, we describe how to unfold a given context-free grammar $G=(\vars,\al,P)$ into a new context-free grammar $G^{\dle{k}}$ which generates exactly the trees of dimension at most $k$ (Definition \ref{def:unf} and Lemma \ref{lem:unf}). Hence, $\amb_{G^{\dle{k}}} \le \amb_{G}$. 
By construction, $G^{\dle{k}}$ is {\em nonexpansive}, i.e.\ every variable $X$ can only be derived into sentential forms in which $X$ occurs at most once \cite{GinsburgSpanier68:deriv-bounded,Yntema:inclusion-relations}.
From this, it easily follows that the commutative ambiguity $\camb_{G^{\dle{k}}}$ is a rational power series in $\Next\aab{\N^\al}$ (Theorem \ref{thm:nonexp}).

We first give an informal description of the notation used in the definiton of $G^{\dle{k}}$:
given the bound $k$ on the maximal dimension we split every variable $X\in\vars$ of $G$ into the variables $X^\deq{d}$ and $X^{\dle{d}}$, where $d\in\{0,1,\ldots,k\}$, with the intended meaning that $X^{\deq{d}}$ resp.\ $X^\dle{d}$ generates all $G_X$-trees of dimension exactly resp.\ at most $d$;
a variable $X^{\dle{d}}$ can only be rewritten to $X^{\deq{d'}}$ for some $d'\le d$, i.e.\ nondeterministically the dimension of the tree to be generated from $X^{\dle{d}}$ has to be chosen; the rules rewriting the variable $X^{\deq{d}}$ 
are derived from the rules $X\to_G \gamma$ by replacing each variable $Y$ occurring in $\gamma$ by either $Y^\deq{d'}$ or $Y^{\dle{d'}}$ 
for some $d'\le d$ in such a way that, inductively, it is guaranteed that every $X$-tree of dimension exactly $d$ is generated exactly once.
In particular, as for each $X$-tree $t=\sigma t_1 \ldots t_r$  there is at most one $i\in [r]$ with $\dim(t) = \dim(t_i)$, the grammar $G^{\dle{k}}$ is nonexpansive.

\begin{definition}\label{def:unf}
Let $G$ be a context-free grammar $G=(\vars,\al,P)$, and let $k$ be a fixed natural number. Set $\vars^\dle{k} := \{ X^\dle{d}, X^\deq{d} \mid X\in\vars, 0 \le d \le k\}$. The grammar $G^{\dle{k}}=(\vars^{\dle{k}},\al,P^{\dle{k}})$ consists then of exactly the following rules:
\begin{itemize}
\item
$\Xdle{d} \to \Xdeq{e}$ for every $d\in [k]\cup\{0\}$, and every $e \in [d]\cup\{0\}$.
\item 
If $X\to_G u_0$, then $X^{\deq{0}}\to_{G^{\dle{k}}} u_0$. 
\item 
If $X\to_G u_0 X_1 u_1$, then $X^{\deq{d}}\to_{G^{\dle{k}}} u_0 X_1^{\deq{d}} u_1$ for every $d\in[k]\cup\{0\}$.
\item 
If $X\to_G u_0 X_1 u_1 \ldots u_{r-1} X_{r} u_r$ with $r>1$:
  \begin{itemize}
  \item
  For every $d\in [k]$, and every $j\in [r]$:
  
  Set $Z_j := X_i^{\deq{d}}$ and $Z_i := X_i^{\dle{d-1}}$ if $i\neq j$ for all $i\in [r]-\{j\}$.
  Then:
  \begin{equation*}
  X^{\deq{d}}\to_{G^{\dle{k}}} u_0 Z_{1} u_1 \ldots u_{r-1} Z_{r} u_r.
  \end{equation*}
  \item
  For every $d\in [k]$, and every $J\subseteq [r]$ with $\abs{J}\ge 2$:
  
  Set 
  $Z_i := X_i^\deq{d-1}$ if $i\in J$ and $Z_i := X_i^\dle{d-2}$ if $i\not \in J$.
  If all $Z_i$ are defined, i.e., $d\ge 2$ if $r>2$, then:
  \begin{equation*}
  X^{\deq{d}}\to_{G^{\dle{k}}} u_0 Z_0 u_1 \ldots  \ldots u_{r-1} Z_{r-1} u_r.
  \end{equation*}
  \end{itemize}\eod
\end{itemize}
\end{definition}

As the sets of variables of $G$ and $G^{\dle{k}}$ are disjoint, in the following, we simply write  $\amb_{X}$ for $\amb_{G,X}$, $\amb_{X^{\dle{d}}}$ for $\amb_{G^{\dle{k}},X^{\dle{d}}}$, $X$-tree for $(G,X)$-tree, and so on.

\begin{lemma}\label{lem:unf}
Every $X^{\deq{d}}$-tree resp.\ $X^\dle{d}$-tree has dimension exactly resp.\ at most $d$. 
There is a yield-preserving bijection between the $X^{\deq{d}}$-trees resp.\ $X^\dle{d}$-trees and the $X$-trees of dimension exactly resp.\ at most $d$.
\end{lemma}

\begin{corollary}\label{cor:amb}
$\amb_{X^{\dle{k}}}(w) = \abs{\{ t \in \tr_{G,X} \mid \Y(t)=w \wedge \dim(t)\le k \}}$ for all $X\in\vars$.\eod
\end{corollary}

\begin{theorem}\label{thm:nonexp}
Let $G=(\vars,\al,P)$ be a context-free grammar.
\begin{enumerate}
\item
$\camb_{X^{\dle{k}}}$ is rational in $\Next\aab{\al^\oplus}$.
\item
There is a $k\in\N$ such that $\amb_{X^{\dle{k}}} = \amb_X$ for all $X\in\vars$ if and only if $G$ is nonexpansive.

Further if such a $k$ exists, then $k < \abs{\vars}$.
Analogously, for $\camb_{X^{\dle{k}}} = \camb_X$.\eod
\end{enumerate}
\end{theorem}
\begin{proof}
The first claim that $\camb_{X^{\dle{k}}}$ is expressible by a weighted rational expression follows directly from the structure of the unfolding of $G^{\dle{k}}$.
With $G^{\dle{k}}$ we associate an algebraic system over $\Next\aab{\N^\al}$ defined by the equations $X = \sum_{X\to \gamma} \gamma$.
The least solution of this system is exactly $\camb$.
For $k=0$ we have only rules which contain at most one variable on the right-hand side. So, the associated algebraic system is linear,
in particular right-linear because of commutativity and, thus, the least solution is expressible by means of a rational expression.
For $k>0$, solving the associated algebraic system bottom up, we have already determined rational expressions for the variables of the form $X^{\dle{d}}$ and $X^{\deq{d}}$
for $d<k$. By the structure of unfolding, the system is again right-linear w.r.t.\ to the remaining variables $X^{\dle{k}}$ and $X^{\deq{k}}$. So the claim follows.

For the second claim, 
assume first that $G$ is expansive. 
Then there is a derivation of the form $Y\Rightarrow w_0 Y w_1 Y w_2$ for some $Y\in\vars$.
Obviously, we can use this derivation to construct $Y$-trees of arbitrary dimension.
Hence, $\camb_{Y^{\dle{k}}} < \camb_Y$ for all $k\in\N$.
Assume now that $G$ is nonexpansive. The definition of ``nonexpansive'' can be restated as:
In any $X$-tree $t=\sigma t_1t_2\ldots t_r$, at most one child contains a node which is labeled by a rule rewriting $X$.
Let $l(t)$ be number of distinct variables $Y$ for which there is at least one node of $t$ which is labeled by a rule rewriting $Y$.
Obviously, $l(t) \le \abs{\vars}$.
Induction on $l(t)$ shows that every derivation tree $t$ satisfying this property has dimension less than $l(t)$:
For $l(t)=1$ a tree with this property cannot contain any nodes of arity two or more. Hence, its dimension is trivially zero.
For $l(t)>1$ given such an $X$-tree $t=\sigma t_1 \ldots t_r$ we can find a simple path $\pi$ leading from the root of $t$ to a leaf
which visits all nodes of $t$ which are labeled by a rule rewriting $X$.
Removing $\pi$ from $t$ we obtain a forest of subtrees each labeled by at most $l(t)-1$ distinct variables,
and each still having above property. Hence, by induction each of these subtrees has dimension less than $l(t)-1$,
and, thus, $t$ has dimension less than $l(t)$.
%
\end{proof}
%
%

We illustrate the construction in the following example.

\begin{example}\label{ex1}
Let $G$ be defined by the productions
\begin{equation*}
X \to a XXXXXX \mid b XXXXX \mid c.
\end{equation*}
The abstract algebraic system associated with this grammar is 
\begin{equation*}
X = aX^6 + bX^5 +c.
\end{equation*}
Using the valuation $h(a) = 1/6$, $h(b)=1/2$, $h(c)=1/3$, we interpret this abstract system as the concrete system
\begin{equation*}
X = 1/6 X^6 + 1/2 X^5 + 1/3
\end{equation*}
over the $\omega$-continuous semiring $\cg{[0,\infty],+,\cdot,0,1}$ of nonnegative reals extended by a greatest element $\infty$ with addition and multiplication extended as in the case of $\Next$.
The least solution $\mu$ of this system, i.e.\ the least nonnegative root of $1/6 X^6 + 1/2 X^5 - X + 1/3$,
can be shown to be neither rational nor expressible using radicals.
We may approximate $\mu$ by evaluating $\camb_{X^{\dle{k}}}$ under $h$.
Up to commutativity, the grammar $G^{\dle{k}}$ corresponds to the following algebraic system:
\begin{equation*}
\begin{array}{lcl@{\hspace{0.5cm}}lcl}
X^{\deq{0}} & = & c &
X^{\dle{0}} & = & c\\[0.2cm]
            & \vdots &  & & \vdots & \\[0.2cm]
X^{\deq{k}} & = & \Bigl( {6 \choose 1} a (X^{\dle{k-1}})^5 + {5 \choose 1 }b (X^\dle{k-1})^4 \Bigr) X^\deq{k} & X^{\dle{k}} & = & \sum_{e=0}^d X^{\deq{e}}\\[0.2cm]
            & + & \sum_{j=2}^6 {6 \choose j} a (X^{\dle{k-2}})^{6-j} (X^\deq{k-1})^j\\[0.2cm]
            & + & \sum_{j=2}^5 {5 \choose j} b (X^{\dle{k-2}})^{5-j} (X^\deq{k-1})^j.\\[0.2cm]
\end{array}
\end{equation*}
From this, rational expressions for $\camb_{X^\dle{k}}$ can easily be obtained:
\begin{equation*}
\begin{array}{lcl@{\hspace{0.5cm}}lcl}
\camb_{X^{\deq{0}}} & = & c &
\camb_{X^{\dle{0}}} & = & c\\[0.2cm]
\camb_{X^{\deq{1}}} & = & (6a c^5 + 5bc^4)^\ast (ac^6 + bc^5) &
\camb_{X^{\dle{1}}} & = & \camb_{X^{\deq{1}}} + \camb_{X^{\dle{0}}}\\[0.2cm]
                    & \vdots & & & \vdots & \\[0.2cm]
\camb_{X^{\deq{k}}} & = & \Bigl( {6 \choose 1} a \camb_{X^{\dle{k-1}}}^5 + {5 \choose 1 }b \camb_{X^{\dle{k-1}}}^4 \Bigr)^\ast &  \camb_{X^{\dle{k}}} & = & \camb_{X^{\deq{k}}} + \camb_{X^{\dle{k-1}}}\\[0.2cm]
            & + & \sum_{j=2}^6 {6 \choose j} a \camb_{X^{\dle{k-2}}}^{6-j} \camb_{X^{\deq{k-1}}}^j\\[0.2cm]
            & + & \sum_{j=2}^5 {5 \choose j} b \camb_{X^{\dle{k-2}}}^{5-j} \camb_{X^{\deq{k-1}}}^j.\\[0.2cm]
\end{array}
\end{equation*}
Evaluating the first three expressions for $\camb_{X^{\dle{k}}}$ under $h$ we obtain the following approximations of $\mu$:
\begin{equation*}
\begin{array}{lcl}
h(\camb_{G^{\dle{k}},X^{\dle{0}}}) & = & 1/3\\[0.1cm]
h(\camb_{G^{\dle{k}},X^{\dle{1}}}) & = & 1/3 + ( 6^{-1} 3^{-6} + 2^{-1} 3^{-5}) (1 - 6\cdot 6^{-1} 3^{-5} - 5\cdot 2^{-1} 3^{-4})^{-1}\\[0.1cm]
                                   & = & \frac{1417}{4221} \approx 0.335702\\[0.1cm]
h(\camb_{G^{\dle{k}},X^{\dle{2}}}) & = & \frac{10981709605561545700033}{32712506178044757018129} \approx 0.335704\\[0.1cm]
\end{array}
\end{equation*}
It can be shown that $h(\camb_{X^{\dle{k}}})$ is exactly the $k$-th approximation obtained by applying Newton's method to $1/6 X^6 + 1/2 X^5 - X + 1/3$ starting at $X=0$.
\eod
\end{example}
\section{Speed of Convergence}\label{sec:conv}
For this section, let $n$ denote the number of variables of the context-free grammar $G$. In \cite{EKL07:stacs} it was shown that,
if $\camb_{X^{\dle{n}}}(\vv) < \camb_{X}(\vv)$, then  $1 \le \camb_{X^{\dle{n}}}(\vv)$, i.e.\ $\supp(\camb_{X^{\dle{n}}})=\supp(\camb_{X})$.
As $\camb_{X^{\dle{n}}}$ is rational, this lower bound yields an alternative proof that $\pa(L(G,X))$ is a regular language.
In this section we extend this result to a lower bound on the speed at which $\camb_{X^{\dle{k}}}$ converges to $\camb_X$ for $k\to \infty$:

By $l(t)$ we denote the number of variables occuring in a derivation tree $t$.
The following lemma was proven in \cite{EKL07:stacs}.
\begin{lemma}
\label{lem:tree-compact}
For every $X$-tree $t$ there is a Parikh-equivalent tree $\tilde{t}$ of dimension at most $l(t)$.
\end{lemma}

By similar arguments as before we can derive an even stronger convergence-theorem:
\begin{theorem}\label{thm:parikh-ext}
Let $n$ be the number of variables of $G$. Then for all $k \geq 0$ and $\vv\in \N^\al$:
$\camb_{G^{\dle{n+k}}}(\vv) \geq \min(\camb_X(\vv), 2^{2^k})$.
\eod
\end{theorem}
\begin{proof}
Assume there is a $\vv\in\N^\al$ with $\camb_{X^{\dle{n+k}}}(\vv) < \camb_X(\vv)$, i.e.\ we have some $X$-tree $t$ of dimension at least $n+k+1$ with $\pa(\Y(t))=\vv$. We show that $t$ witnesses the existence of at least $2^{2^k}$ distinct $X$-trees of dimension at most $n+k$ with a yield that is Parikh-equivalent to $t$.

We will prove the following stronger statement which implies the statement of the theorem:
If $\dim(t) \geq l(t) + k + 1$ then there exist at least $2^{2^k}$ Parikh-equivalent trees of dimension at most
$l(t) + k$.

We prove the claim by induction on $|V(t)|$, the number of nodes of $t$.
If $|V(t)| = 1$, then $\dim(t) = 0$ whereas $l(t)+k+1 = k+2 > 0$, so the claim trivially holds.
Observe that if $t$ has a subtree of dimension at least $l(t) + k + 1$
we can apply the induction hypothesis to every such subtree and thus obtain altogether at least $2^{2^k}$ Parikh-equivalent 
trees of dimension lower than $\dim(t)$.
Therefore we can restrict ourselves to the case where $\dim(t) = l(t) + k + 1$ and all subtrees have dimension at most
$l(t) + k$. Note that in this case $t$ must have (at least) two subtrees $t_1,t_2$ of dimension exactly $l(t) + k$.
We distinguish two cases:
\begin{itemize}
\item Case $l(t_1)<l(t)$ or $l(t_2)<l(t)$: Suppose w.l.o.g.\ $l(t_1)<l(t)$. Apply the induction hypothesis to $t_1$, 
since $\dim(t_1) = l(t) + k \geq l(t_1) + k + 1$ and obtain at least $2^{2^k}$ Parikh-equivalent trees of dimension at most 
$l(t_1)+k$. Then we apply Lemma \ref{lem:tree-compact} to every \emph{other} subtree of $t$ to obtain at least $2^{2^k}$ 
different trees $\tilde{t}$ of dimension at most $l(t)+k$.
\item Case $l(t_1)=l(t_2)=l(t)$: (This is the only case that requires actual work)
Since $t_1$ has dimension $l(t) + k$ it contains a perfect binary tree of height $l(t) + k$ as a minor. The set of nodes of 
this minor on level $k$ define $2^k$ (independent) subtrees of $t_1$. Each of these $2^k$ subtrees has height at least 
$l(t)$, thus by the Pigeonhole principle contains a path with two variables repeating.
We reallocate any subset of these $2^k$ pump-trees to $t_2$ which is possible since $l(t_2)=l(t)=l(t_1)$. This changes the 
subtrees $t_1,t_2$ into $\tilde{t_1},\tilde{t_2}$.
Each of these $2^{2^k}$ choices produces a different tree $\tilde{t}$---the trees differ in the subtree $\tilde{t_1}$.
As in the previous case we now apply Lemma \ref{lem:tree-compact} to every subtree of $t$ except $t_1$ thereby reducing the dimension of $\tilde{t}$ to at most $\dim(t_1) = l(t) + k$ thus obtaining
at least $2^{2^k}$ different Parikh-equivalent trees of dimension at most $\dim(t_1) = l(t) + k$.
\end{itemize}
\end{proof}

We state some straightforward consequences of Theorem~\ref{thm:parikh-ext} based on the generalization of context-free grammars to algebraic systems.
We say that a $\omega$-continuous semirng $S$ is {\em collapsed} at some positive integer $k$ if in $S$ the identity $k=k+1$ holds. For instance, the semirings $\N_k\aab{\al^\ast}$ and $\N_k\aab{\N^\al}$ are collapsed at $k$. For $k=1$, the semiring is idempotent. 

\begin{corollary}\label{cor:camb}
$\camb_{X^{\dle{n+\log\log k}}} = \camb_{X}$ over $\N_k\aab{\N^\al}$, and $\camb_{X}$ is rational in $\N_k\aab{\N^\al}$.
\end{corollary}

\begin{corollary}
The least solution of an algebraic system with associated context-free grammar $G$ and valuation $h$ over a commutative $\omega$-continuous semiring $S$ collapsed at $k$ is $(h(\camb_{X^{\dle{n+\log\log k}}})\mid X\in\vars)$.
\end{corollary}

By the results of \cite{EKL07:stacs}, the latter corollary is equivalent to saying that Newton's method reaches the least solution of an algebraic system in $n$ variables over a commutative $\omega$-continuous semiring collapsed at $k$ after at most $n+\log\log k$ steps.

\section{Semilinearity}\label{sec:semilinear}
In the following, let $k$ denote a fixed positive integer. 
By Corollary \ref{cor:camb} we know that $\camb_G$ is rational modulo $k=k+1$. In this section, we give a semilinear characterization also of $\camb_G$. We identify in the following a word $w\in\al^\ast$ with its Parikh vector $\pa(w)\in\N^\al$.

In the idempotent setting ($k=1$), see e.g.\ \cite{Pil73,DBLP:books/sp/KuichS86,HK99,AEIparikh}, the identities (i) $(x^\ast)^\ast = x^\ast$, (ii) $(x+y)^\ast = x^\ast y^\ast$, and (iii) $(xy^\ast)^\ast = 1 + xx^\ast y^\ast$ can be used to transform any regular expression into a regular expression in ``semilinear normal form'' $\sum_{i=1}^r w_{i,0} w_{i,1}^\ast \ldots w_{i,l_r}^\ast$ with $w_{i,j}\in\al^\ast$. 
It is not hard to deduce the following identities over $\N_k\aab{\N^\al}$ where $x^{<r}$ abbreviates the sum $\sum_{i=0}^{r-1} x^i$
and $\supp(x)$ is identified with its characteristic function:

\begin{lemma}\label{lem:id}
The following identities hold over $\N_k\aab{\N^\al}$:
\begin{equation*}
\begin{array}{c@{\hspace{0.5cm}}lcl}
\mathsf{(I1)} & k{x} & = & k \supp({x})\\[0.2cm]
\mathsf{(I2)} & (\gamma {x})^\ast & = & (\gamma{x})^{< \lceil \log_{\gamma} k\rceil} + k {x}^{\lceil \log_\gamma k\rceil} {x}^\ast\\[0.2cm]
\mathsf{(I3)} & ({x}^\ast)^\ast & = & k {x}^\ast\\[0.2cm]
\mathsf{(I4)} & ({x}+{y})^\ast & = & ({x}+{y})^{< k} + {x}^k {x}^\ast + {y}^k {y}^\ast + k {x} {y} ({x}+{y})^{\max(k-2,0)} {x}^\ast {y}^\ast\\[0.2cm]
\mathsf{(I5)} & ({x}{y}^\ast)^\ast & = & 1 + xy^\ast + x^2 x^\ast + x^2 y \sum_{0\le m,j < k-2} {2+m+j\choose 1+j} x^m y^j\\[0.2cm]
                     &  & + & k x^2 y (x^{\max(k-2,0)} + y^{\max(k-2,0)}) x^\ast y^\ast
\end{array}
\end{equation*}
for $\gamma$ any integer greater than one.\eod
\end{lemma}

Consider a rational series $\mf{r}\in \N_k\aab{\N^\al}$ represented by the rational expression $\rho$. 
The above identities, where $\mathsf{(I3),(I4),(I5)}$ generalizes (i), (ii), (iii), respectively, allow one to reduce the star height of $\rho$ to at most one by distributing the Kleene stars over sums $(\rho_1 +\rho_2)^\ast$ and products $(\rho_1 \rho_2)^\ast$ -- in the latter case if $\rho_1\rho_2 \not\in\al^\ast$ --
yielding a rational expression $\rho'$ of the form
\begin{equation*}
\rho'=\sum_{i=1}^s \gamma_i w_{i,0} w_{i,1}^\ast \ldots w_{i,l_i}^\ast \quad (w_{i,j}\in\al^\ast,\ \gamma_{i}\in\N_k).
\end{equation*}
which still represents $\mf{r}$ over $\N_k\aab{\N^\al}$. 
By $\mathsf{(I1)}$ we know that, if $\gamma_{i,0}=k$, we may replace $w_{i,0} w_{i,1}^\ast \ldots w_{i,l_i}^\ast$ by its support which is a linear set in $\N^\al$. This can be generalized to $k>1$:

\begin{theorem}\label{thm:semilin}
Every rational $\mf{r}\in \N_k\aab{\N^\al}$ can be represented as a finite sum of weighted linear sets, i.e.\
\begin{equation*}
\mf{r} = \sum_{i\in[s]} \gamma_i \supp(w_{i,0} w_{i,1}^\ast \ldots w_{i,l}^\ast)\ \text{ with $w_{i,j}\in\al^\ast$ and $\gamma_i\in \N_k$}.
\end{equation*}
\end{theorem}

\begin{example}
The rational expression $\rho=(a+2b)^\ast$ represents the series $\sum_{i,j\in\N} 2^j a^i b^j$ in $\Next\aab{\N^\al}$.
Computing over $N_2\aab{\N^\al}$ we may transform $\rho$ as follows:
\begin{equation*}
\begin{array}{cl@{\hspace{1cm}}l}
  &(a+2b)^\ast & \mathsf{(I4)} \\
= & (a+2b)^{<2} + a^2 a^\ast + (2b)^2 (2b)^\ast + 2 a (2b) a^\ast (2b)^\ast & (\mathsf{I1})\\
= & \ew + a + 2b + a^2 a^\ast + 2 b^2 b^\ast + 2 ab a^\ast b^\ast & (x^\ast = \sum_{i\in\N} x^i,\mathsf{I1})\\
= & a^\ast + 2( b b^\ast + ab a^\ast b^\ast) & (x^\ast = \sum_{i\in\N} x^i,\mathsf{I1})\\
= & a^\ast + 2( b b^\ast a^\ast ) &  \mathsf{(I1)}\\
= & a^\ast + 2 \supp( b b^\ast a^\ast) & (a^\ast = \sum_{i\in\N} 1 a^i)\\
= & 1 \supp(a^\ast) + 2 \supp( b b^\ast a^\ast).
\end{array}
\end{equation*}
\end{example}

\begin{corollary}\label{cor:presb}
For every $k\in\Next$ we can construct a formula of Presburger arithmetic that represents the set $\{ \vv \in \N^\al\mid \camb_{G,X}(\vv) = k \}$.
\end{corollary}

\section{Acknowledgment}
The authors likes to thank Volker Diekert for his help with a first version of Theorem~\ref{thm:parikh-ext},
Rupak Majumdar for his pointer to \cite{DBLP:conf/pods/GreenKT07}, and
Javier Esparza and Andreas Gaiser for many helpful discussions.

\bibliographystyle{alpha}
\bibliography{db}

\newpage

\appendix
\section{Missing proofs}

\newcommand{\diml}{\textsf{dlen}}
\newcommand{\dimc}{\textsf{dchar}}
\subsection*{Proof of Lemma \ref{lem:unf}}
Let $t$ be a derivation tree of dimension $\dim(t)=d$. Then $t=\sigma t_1\ldots t_r$ has at most one child $t_c$ ($c\in[r]$) with $\dim(t)=\dim(t_c)$ by definition of $\dim$. Hence, there is a unique maximal path $v_0\ldots v_l$ starting in $v_1=\ew$
such that (i) $\dim(t)=\dim(t|_{v_l})$ and (ii) either $v_l$ is a leaf of $t$ or every proper subtree of $v_l$ has dimension less than $d$.
Let $\diml(t)=l$ denote the length of this unique path. 
Further, we use $\dimc(t) = \{ (i,\dim(t'_i)) \mid i \in [r'] \text{ for } t|_{v_l} = \sigma' t'_1 \ldots t'_{r'} \}$ to remember the dimensions of the children of $t|_{v_l}$. ($\dimc(t)=\emptyset$ if $v_l$ is a leaf of $t$.)

We first construct a mapping $\hat{\cdot}$ from the derivation trees of $G^{\dle{k}}$ to the derivation trees of $G$ of dimension at most $d$ and exactly $d$, respectively: 
\begin{itemize}
\item If $t=\sigma_{X^\dle{d},X^\deq{e}} t_1$, then $\widehat{t} := \widehat{t_1}$.
\item If $t=\sigma_{X^\deq{d},u_0 Z_1 u_1 \ldots u_{r-1} Z_{r} u_r} t_1 \ldots t_r$, then $\widehat{t} := \sigma_{X,u_0 X_1 u_1 \ldots u_{r-1} X_{r} u_r} \widehat{t_1} \ldots \widehat{t_r}$ where $X_i\in\vars$ is the variable from which $Z_i\in\vars^{\dle{k}}$ was derived.
\end{itemize}
Informally, $\hat{\cdot}$ contracts edges induced by rules $X^{\dle{d}}\to X^\deq{e}$ which choose a concrete dimension $e\le d$,
and then forgets the superscripts. By definition, the rules of $G^{\dle{k}}$ which rewrite the variable $X^{\deq{d}}$ are obtained from the rules of $G$ which rewrite the variable $X$ by only adding superscripts. Hence, $\hat{\cdot}$ maps any $X^{\dle{d}}$-tree and any $X^{\deq{d}}$-tree to a $X$-tree while preserving its yield ($\Y(t)=\Y(\hat{t})$). Further, as the edges induced by the rules $X^{\dle{d}}\to X^{\deq{e}}$ do not influence the tree dimension, we also have $\dim(t) = \dim(\hat{t})$ and $\dimc(t)=\dimc(\hat{t})$. We also have $\diml(t)\ge \diml(\hat{t})$ as contracting the edges induced by $X^{\dle{d}}\to X^{\deq{e}}$ can only reduce $\diml(\cdot)$.

We claim  that $\hat{\cdot}$ maps the set of $X^{\dle{d}}$-trees ($X^{\deq{d}}$) one-to-one onto the set of $X$-trees of dimension at most $d$ (exactly $d$). We proceed by induction on $d$. Let $d=0$.
  \begin{itemize}
  \item 
  Consider a $X^{\deq{0}}$-tree $t$. The only rules rewriting $X^{\deq{0}}$ are of the form $X^{\deq{0}}\to u$ or $X^{\deq{0}}\to u Y^{\deq{0}} v$ (for $u,v\in\al^\ast$ and $Y\in\vars$). For these rules, forgetting the superscript is an injective operation. Hence, $\hat{\cdot}$ is injective on the set of $X^{\deq{0}}$-trees. Obviously, $t$ is also a chain, and, thus, $0=\dim(t)=\dim(\hat{t})$. (In fact, $\diml(t)=\diml(\hat{t})$.)
  
  Consider now a $X^{\deq{0}}$-tree $t$. By definition of $G^{\dle{k}}$, $X^{\dle{0}}$ can only be rewritten to $X^{\deq{0}}$. So $t=\sigma_{X^{\dle{0}},X^{\deq{0}}} t_1$ for $t_1$ a $X^{\deq{0}}$-tree, and $\hat{t}=\hat{t}_1$. Again, $0=\dim(t)=\dim(\hat{t})$.
  \item
  Let $t$ be a $X^{\deq{d}}$-tree for $d>0$ where $t= \sigma_{X^{\deq{d}},u_0 Z_1 u_1 \ldots u_{r-1} Z_r u_r} t_1\ldots t_r$ for some $r>0$
  where there is a rule $X\to u_0 X_1 u_1 \ldots u_{r-1} X_r u_r$ in $G$ ($X_i\in\vars$, $u_i\in\al^\ast$) such that for all $i\in[r]$ either $Z_i \in \{ X_i^{\deq{d}}, X_i^{\dle{d-1}}\}$ or $Z_i \in \{ X_i^{\deq{d-1}}, X_i^{\dle{d-2}}\}$.
  
  Assume first that $t$ has no $Y^{\dle{d}}$-subtree for any $Y\in\vars$, i.e.\ $t$ is a $X^{\deq{d}}$-tree of minimal height. Then $\sigma=\sigma_{X^{\deq{d}},u_0 Z_1\ldots Z_r u_r}$ where $Z_i = X_i^{\deq{d-1}}$ or, if $d\ge 2$, $Z_i = X_i^{\dle{d-2}}$ for some $X_i\in\vars$ such that $X\to u_0 X_1\ldots X_r u_r$ in $G$. 
  Inductively, we already know that $\dim(t') = e$ ($\dim(t')\le e$) for every $X^{\deq{e}}$-tree ($X^{\dle{e}}$-tree)
  and all $e<d$. Hence, $\dim(t)=\dim(\hat{d})=d$ and $\diml(t)=\diml(\hat{t})=0$.
  
  Thus, assume that $t$ contains a $Y^{\deq{d}}$-subtree for some $Y\in\vars$. By construction, there occurs at most one ``$\deq{d}$-variable'', i.e.\ a variable of $\{ Y^{\deq{d}} \mid Y\in\vars\}$, in the right-hand side $\gamma$ of every rule $X^{\deq{d}}\to \gamma$. 
  By construction, there is a unique $j\in[r]$ such that $Z_j = X_j^{\deq{d}}$, while $Z_i = X_i^{\dle{d-1}}$ for all $i\in [r]-\{j\}$.
  Then the $X_j^{\deq{d}}$-tree $t_j$ has height less than $t$, so by induction on the height of $X^{\deq{d}}$-trees, we have $\dim(t_j)=\dim(\hat{t}_j)=d$ and $\diml(t_j)=\diml(\hat{t}_j)$.
  By induction on $d$, we also know that $\dim(t_i) < d$ ($i\in [r]-\{j\}$). Hence, $\dim(t)=d$ and $\diml(t)=\diml(t_j)+1$. As the edge to $t_j$ is not contracted by $\hat{\cdot}$, also $\diml(t)=\diml(t_j)+1=\diml(\hat{t}_j)+1=\diml(\hat{t})$.
  
  Assume now that $\hat{t}=\hat{t}'$ for two $X^{\deq{d}}$-trees $t,t'$. Then $\dim(t)=\dim(t')=\dim(\hat{t})$, $\diml(t)=\diml(t')=\dim(\hat{t})$, and $\dimc(t)=\dimc(t')=\dimc(\hat{t})$. 
  Let $\hat{t}=\sigma_{X,u_0 X_1 u_1 \ldots u_{r-1} X_r u_r}$.
  Then necessarily, $t=\sigma_{X^{\deq{d}},u_0 Z_1 \ldots Z_r u_r}$ and $t' = \sigma_{X^{\deq{d}},u_0 Z_1' \ldots Z'_r u_r}$ 
  with either $Z_i \in \{ X_i^{\deq{d}}, X_i^{\dle{d-1}}\}$ or $Z_i \in \{ X_i^{\deq{d-1}}, X_i^{\dle{d-2}}\}$, and, analogously, for all $Z'_i$.
  as $\hat{\cdot}$
  only forgets superscripts and removes $\sigma_{X^{\dle{d}},X^{\deq{e}}}$.
  
  If $\diml(\hat{t})=0$, then $t,t',\hat{t}$ have only subtrees of dimension at most $d-1$. By definition of $G^{\dle{k}}$, 
  it follows that $Z_i,Z_i' \in \{ X_i^{\deq{d-1}}, X_i^{\dle{d-2}}\}$.
  By induction, we know that only $\deq{d-1}$-variables can generate trees of dimension $d-1$, hence, necessarily $Z_i=Z'_i=X_i^{\deq{d-1}}$ for all children $i\in[r]$ of $\hat{t}$ which have dimension exactly $d-1$, while $Z_i = Z'_i = X_i^{\dle{d-2}}$ for all remaining children.
  Again by induction, we know that $\hat{\cdot}$ is injective on sets of $Y^{\dle{d-2}}$-trees and $Y^{\deq{d-1}}$-trees, respectively.
  Hence, $t=t'$.
  
  Finally, assume $\diml(\hat{t})> 0$. Then $\hat{t}$ has a unique child $t|_j$ of dimension $d$, while $\dim(t|_i)< d$ for $j\in [r]-\{i\}$.   Consequently, $Z_j = Z'_j = X_j^{\deq{d}}$ and $Z_j=Z_j'=X_j^{\dle{d-1}}$ for $j\in[r]-\{i\}$ by definition of $G^{\dle{k}}$. By induction on $d$ and $\diml(t)$, we may assume that $\hat{\cdot}$ is injective on the subtrees of $t$ and $t'$, hence, $t=t'$ follows.
  \end{itemize}
It remains to show that for any $X$-tree $t'$ of dimension exactly $d$ (at most $d$), there is a $X^{\deq{d}}$-tree ($X^{\dle{d}}$-tree) $t$
such that $\hat{t}=t'$. To this end, we define an operator $\check{\cdot}$ which maps a $X$-tree of dimension exactly $d$ to a $X^{\deq{d}}$-tree by, essentially, introducing the superscripts into a symbol $\sigma_{X,u_0 X_1\ldots X_r u_r}$ as required by the dimensions of the subtrees $t_1,\ldots,t_r$:
\begin{quote}
Let $t=\sigma_{X,u_0 X_1 u_1 \ldots X_r u_r} t_1 \ldots t_r$ with $d=\dim(t)$ and $d_i = \dim(t_i)$, then 
\begin{equation*}
\check{t} := \sigma_{X^{\dle{k}},X^{\deq{d}}} \sigma_{X^\deq{d},u_0 Z_1 u_1 \ldots Z_r u_r} t'_1 \ldots t'_r.
\end{equation*}
where $Z_i,t'_i$ are defined as follows:
\begin{itemize}
\item
If $d > \max_{i\in[r]} d_i$, then let $J = \{ i \in [r] \mid d_i = d-1\}$ and set $Z_i := X_i^{\deq{d-1}}$ and $t'_i := \check{t_i}$ if $i\in J$, and $Z_i := X_i^{\dle{d-2}}$ and $t'_i := \sigma_{X^{\dle{d-2}},X^{\deq{d_i}}} \check{t_i}$ otherwise.
\item If $d = \max_{i\in[r]} d_i$, then there is a unique $j\in [r]$ such that $d_j = d$.
Set $Z_j = X_j^{\deq{d}}$ and $t'_j := \check{t}_j$. For the remaining $i\in [r]-\{j\}$, set $Z_i := X_i^{\dle{d-1}}$ and $t'_i := \sigma_{X_i^{\dle{d-1}},X_i^\deq{d_i}} \check{t}_i$.
\end{itemize}
\end{quote}
It is straightforward to check that $\check{t}$ is indeed a $X^{\deq{d}}$-tree for $\dim(t)=d$, and that $\hat{\check{t}}=t$.
Obviously, $\check{\cdot}$ is injective.
Finally, for every $d'\ge d$ there is exactly one rule $X^{\dle{d'}}\to X^{\deq{d}}$. Hence, $\sigma_{X^{\dle{d'}},X^{\deq{d}}}\check{t}$ is, by definition of $G^{\dle{k}}$, the unique $X^{\dle{d'}}$-tree which is mapped by $\widehat{\cdot}$ back to $t$.

\subsection*{Proof of Lemma \ref{lem:id}}
The proofs are straightforward, and essentially only require to unroll and cut off the power series underlying the Kleene star using the $\omega$-continuity of the Kleene star and the assumption that $k=k+1$.
We several times make use of the trivial bound  
${ a \choose b } \ge a$ for $0 < b < a$.
on the binomial coefficient.
\begin{itemize}
\item[$\mathsf{(I1)}$] $kx = k\supp(x)$ is obviously true modulo $k=k+1$. 

\item[$\mathsf{(I2)}$] $(\gamma {x})^\ast = (\gamma{x})^{< \lceil \log_{\gamma} k\rceil} + k \cdot {x}^{\lceil \log_\gamma k\rceil} {x}^\ast$

This follows from the $\omega$-continuity of the star
$(\gamma x)^\ast = \sum_{n\in\N} (\gamma x)^n$
and the first identity.
\item[$\mathsf{(I3)}$] $(x^\ast)^\ast = k x^\ast$

Choose any $w\in \supp( (x^\ast)^\ast)$. Then $w$ can be factorized into
$w=u_1\ldots u_l$ with $u_i \in \supp(x^\ast)$, i.e., $w\in \supp( (x^\ast)^l)$. Obviously, we then can also find a factorization of $w$ into $l+i$ words for any $i>0$ as we may add an arbitrary number of neutral elements $\ew$ into this factorization. Hence, $w\in \supp((x^\ast)^{l+i})$ for all $i\ge 0$. So, the coefficient of $w$ in $(x^\ast)^\ast$ is $\infty = k$ modulo $k=k+1$.
\item[$\mathsf{(I4)}$]
$({x}+{y})^\ast =  ({x}+{y})^{< k} + {x}^k {x}^\ast + {y}^k {y}^\ast + k {x} {y} ({x}+{y})^{\max(k-2,0)} {x}^\ast {y}^\ast$

Proof:
\begin{equation*}
\begin{array}{l@{\hspace{0.5cm}}cl}
         &   & (x+y)^\ast\\[0.2cm]
         & = & (x+y)^{<k} + \sum_{n\ge k}(x+y)^n \\[0.2cm]
(xy=yx)  & = & (x+y)^{<k} + \sum_{n\ge k} \sum_{j=0}^n { n\choose j} x^j y^{n-j} \\[0.2cm]
         & = & (x+y)^{<k} + \sum_{n\ge k} x^n + y^n + \sum_{j=1}^{n-1} { n\choose j} x^j y^{n-j}) \\[0.2cm]
         & = & (x+y)^{<k} + x^k x^\ast + y^k y^\ast
          +   \sum_{n\ge k} \sum_{j=1}^{n-1} { n\choose j} x^j y^{n-j} \\[0.2cm]
(j=i+1,n=m+2)
         & = & (x+y)^{<k} + x^k x^\ast + y^k y^\ast\\
         & + &  \sum_{m\ge \max(k-2,0)} \sum_{i=0}^{m} { m+2\choose i+1} x^{i+1} y^{m-i+1} \\[0.2cm]         
({m+2\choose i+1}\ge k,\mathsf{(I1)})         
         & = & (x+y)^{<k} + x^k x^\ast + y^k y^\ast\\
         & + &  k xy \sum_{m\ge \max(k-2,0)} \sum_{i=0}^{m} { m\choose i} x^{i} y^{m-i} \\[0.2cm]         
         & = & (x+y)^{<k} + x^k x^\ast + y^k y^\ast\\
         & + &  k xy (x+y)^{\max(k-2,0)} (x+y)^\ast\\[0.2cm]                  
((ii)\ \supp((x+y)^\ast)=\supp(x^\ast y^\ast), \mathsf{(I1)})
         & = & (x+y)^{<k} + x^k x^\ast + y^k y^\ast\\
         & + &  k xy (x+y)^{\max(k-2,0)} x^\ast y^\ast\\[0.2cm]                  
\end{array}         
\end{equation*}
\item[$\mathsf{(I5)}$]
\begin{equation*}
\begin{array}{lcl}
({x}{y}^\ast)^\ast & = & 1 + xy^\ast + x^2 x^\ast
                                             + x^2 y \sum_{0\le m,j < k-2} {2+m+j\choose 1+j} x^m y^j\\[0.2cm]
                                             & + & k x^2 y x^{\max(k-2,0)} x^\ast y^\ast
                       + k x^2 y x^\ast y^{\max(k-2,0)} y^\ast
\end{array}
\end{equation*}
Proof:
\begin{equation*}
\begin{array}{l@{\hspace{0.5cm}}cl}
                       &   & (xy^\ast)^\ast\\[0.2cm]
(xy^\ast = y^\ast x)   & = & \sum_{n\in\N} x^n (y^\ast)^n\\[0.2cm]
                       & = & 1 + xy^\ast\\
                       & + & \sum_{n\ge 2} x^n \sum_{l\ge 0} { n+l-1 \choose l} y^l\\[0.2cm]
                       & = & 1 + xy^\ast + x^2 x^\ast\\
                       & + & \sum_{n\ge 2, l\ge 1} {n+l-1\choose l} x^n y^l\\[0.2cm]
(n=m+2,l=j+1,xy=yx)    & = & 1 + xy^\ast + x^2 x^\ast\\
                       & + & x^2 y \sum_{m\ge 0, j\ge 0} {2+m+j\choose 1+j} x^m y^j\\[0.2cm]
                       & = & 1 + xy^\ast + x^2 x^\ast\\
                       & + & x^2 y \sum_{\substack{m,j\ge 0\\ m\ge k-2 \vee j \ge k-2}} {2+m+j\choose 1+j} x^m y^j\\
                       & + & x^2 y \sum_{0\le m,j < k-2} {2+m+j\choose 1+j} x^m y^j\\[0.2cm]
(k=k+1)                & = & 1 + xy^\ast + x^2 x^\ast\\
                       & + & k x^2 y \sum_{\substack{m,j\ge 0\\ m\ge k-2 \vee j \ge k-2}} x^m y^j\\
                       & + & x^2 y \sum_{0\le m,j < k-2} {2+m+j\choose 1+j} x^m y^j\\[0.2cm]
(I1   )                & = & 1 + xy^\ast + x^2 x^\ast\\
                       & + & k x^2 y x^{\max(k-2,0)} x^\ast y^\ast\\
                       & + & k x^2 y x^\ast y^{\max(k-2,0)} y^\ast\\
                       & + & x^2 y \sum_{0\le m,j < k-2} {2+m+j\choose 1+j} x^m y^j\\[0.2cm]                       
\end{array}         
\end{equation*}
\end{itemize}

\subsection*{Proof of Theorem~\ref{thm:semilin}}

We  identify a word $w\in\al^\ast$ with its Parikh vector $\pa(w)\in\N^\al$. 
We show that, if $\supp(w_1^\ast \ldots w_l^\ast)\neq w_1^\ast \ldots w_l^\ast$ in $\N_k\aab{\N^\al}$,
then we can split the linear term in a finite sum of weighted linear terms where in each linear term with weight less than $k$ the number of Kleene stars is strictly less than $l$. Then the result follows inductively. 

W.l.o.g.\ we may assume that each $w_i\neq \ew$, i.e.\ $\pa(w_i)\neq\vecn{0}$, as $\ew^\ast = \infty = k$.
Denote by $M\in \N^{\al\times l}$ the matrix whose $i$-th row is given by $\pa(w_i)$ (w.r.t.\ some chosen order on $\al$),
and let $\vecn{\lambda}=(\lambda_1,\ldots,\lambda_l)\in\N^l$.
Then the coefficient $c_{\vv} := (w_1^\ast \ldots w_l^\ast, \vv)$ is exactly the number of solutions over $\N^l$ of the linear equation $\vv = \vecn{\lambda} M$.
If the set $\{\pa(w_1),\pa(w_2),\ldots,\pa(w_l)\}$ is linearly independent, then trivially $c_{\vv} \le 1$ and we are done.

Assume thus that the set $\{\pa(w_1),\pa(w_2),\ldots,\pa(w_l)\}$ is linearly dependent, i.e.\
there is some kernel vector $\vecn{n}=(n_1,\ldots,n_l)\in\Z^l\setminus\{\vecn{0}\}$.
Let $I_{+} = \{ i \in [l] \mid n_i > 0\}$, $I_{-} = \{ i \in [l] \mid n_{i} < 0\}$, and $I_0 = \{ i\in [l] \mid n_{i} = 0\}$.
As all components of $M$ are nonnegative, $\vecn{n}$ necessarily has a positive and a negative component, i.e.\ $I_+ \neq \emptyset \neq I_{-}$. 
Let $\norm{n}_{\infty}:=\max_{i\in[l]} \abs{n_i}$ and $C:= \norm{n}_{\infty} \cdot (k-1)$.

Consider now any $\vecn{\lambda}=(\lambda_1,\ldots,\lambda_l)\in \N^l$ with $\lambda_i > C$ for all $i\in I_{+}$. 
Then also $\vecn{\lambda}-i\vecn{n} \in \N^l$ for $i=0,\ldots,k-1$ and trivially $\vv = \vecn{\lambda} M = (\vecn{\lambda}-i\vecn{n})M$ which implies that $c_{\vv} \ge k$.
If $\lambda_i > C$ for all $i\in I_-$, consider analogously $\vecn{\lambda}+i\vecn{n}$.
For $I\in\{ I_+,I_-\}$ we split the series $\prod_{i\in I} w_i^\ast$ into series $\mf{s}_I$ and $\mf{t}_I$ defined by
\begin{equation*}
\mf{s}_I := \prod_{i\in I} (w_i^C w_i^\ast)\, \text{ and }\, \mf{t}_{I} := \sum_{\emptyset\neq J \subseteq I} \prod_{i\in J} w_i^{<C} \prod_{i\in I-J} (w_i^C w_i^\ast)
\end{equation*}
As discussed above, all positive coefficients of $\mf{s}=\prod_{i\in I} (w_i^C w_i^\ast)$ (for $I\in \{I_+,I_-\}$) are greater than or equal to $k$. Hence  $\mf{s}_I = k \supp(\mf{s}_I)$ over $\N_k\aab{\N^\al}$.
\begin{equation*}
\begin{array}{cl}
  & w_1^\ast w_2^\ast \ldots w_l^\ast \\[0.2cm]
= & \displaystyle \prod_{i\in I_0} w_i^\ast (\mf{s}_{I_+} + \mf{t}_+) (\mf{s}_{I_-} + \mf{t}_{I_-})\\[0.2cm]
= & \displaystyle \prod_{i\in I_0} w_i^\ast (k\mf{s}_{I_+} + \mf{t}_+) (k\mf{s}_{I_-} + \mf{t}_{I_-})\\[0.2cm]
= & \displaystyle \prod_{i\in I_0} w_i^\ast \left( \mf{t}_{I_+} \mf{t}_{I_-} + k(\mf{t}_{I_+} \mf{s}_{I_-} + \mf{t}_{I_-} \mf{s}_{I_+} + \mf{s}_{I_-} \mf{s}_{I_+}) \right)\\[0.2cm]
= & \displaystyle \prod_{i\in I_0} w_i^\ast \left( \mf{t}_{I_+} \mf{t}_{I_-} + k(\mf{t}_{I_+} \mf{s}_{I_-} + \mf{t}_{I_-} \mf{s}_{I_+} + 2\mf{s}_{I_-} \mf{s}_{I_+}) \right)\\[0.2cm]
= & \displaystyle \prod_{i\in I_0} w_i^\ast \left( \mf{t}_{I_+} \mf{t}_{I_-} + k \mf{s}_{I_+} \prod_{i\in I_-} w_i^\ast + k \mf{s}_{I-} \prod_{i\in I_+} w_i^\ast \right)\\[0.2cm]
= & \displaystyle \prod_{i\in I_0} w_i^\ast \left( \mf{t}_{I_+} \mf{t}_{I_-} + k \left( \prod_{i\in I_+} w_i^C + \prod_{i\in I_-} w_i^C\right) \prod_{i\in I_+ \cup I_-} w_i^\ast \right)\\[0.2cm]
= & \displaystyle k \left( \prod_{i\in I_+} w_i^C + \prod_{i\in I_-} w_i^C\right) \prod_{i\in [l]} w_i^\ast  +  \mf{t}_{I_+} \mf{t}_{I_-} \prod_{i\in I_0} w_i^\ast\\[0.2cm]
\end{array}
\end{equation*}
It remains to consider the second summand which can be written as a finite sum of products of which each contains at most $\abs{[l] - (J_+ \cup J_-)} \le l-2$ Kleene stars:
\begin{equation*}
\mf{t}_{I_+} \mf{t}_{I_-} \prod_{i\in I_0} w_i^\ast = \sum_{\substack{\emptyset \neq J_+ \subseteq I_+\\ \emptyset\neq J_{-} \subseteq I_{-}}} \prod_{i\in J_{+}\cup J_{-}} w_i^{<C} \prod_{i\in (I_+-J_+)\cup (I_- - J_-)} w_i^C \prod_{i\in [l] - (J_+ \cup J_-)} w_i^\ast.
\end{equation*}

\subsection*{Proof of Corollary~\ref{cor:presb}}
As $\pa(L(G,X))=\supp( \camb_{G,X} ) = \{ \vv \in \N^\al \mid \camb_{G,X}(\vv) > 0\}$ is semilinear by Parikh's theorem, it is effectively representable by a formula of Presburger arithmetic, and so is its complement ($k=0$).

Assume thus $1 \le k < \infty$ and let $K=k+1$. Then we may compute from $\camb_{X^{\dle{n+\log\log K}}}$ a weighted semilinear representation of $\camb_{X}$ modulo $K=K+1$:
\begin{equation*}
\camb_{X} = \sum_{i=1}^r \gamma_i \supp( \vv_{i,0} \vv_{i,1}^\ast \ldots \vv_{i,l_i}^\ast )\ \text{ with $\gamma_i\in \N_K$ and $\vv_{i,j} \in \N^\al$}.
\end{equation*}
From each term $\supp( \vv_{i,0} \vv_{i,1}^\ast \ldots \vv_{i,l_i}^\ast )$ we can construct an equivalent Presburger formula $F_i$.
Then $\camb_{X}(\vv) = k$ if and only if 
\begin{equation*}
\vv \models \exists y_1,\ldots,y_r\colon \sum_{i=1}^r \gamma_i y_i = k \wedge \bigwedge_{i=1}^l (F_i(\vv)\rightarrow y_i = 1 \wedge \neg F_i(\vv) \rightarrow y_i = 0 ).
\end{equation*}
Finally, let $k=\infty$. As for any $\vv\in\N^\al$ there are only finitely many $w\in\al^\ast$ with $\pa(w)=\vv$,
we have $\camb_{G,X}(\vv)=\infty$ if and only if there is a $w\in\al^\ast$ with $\pa(w)=\vv$ and $\amb_{G,X}(w)=\infty$.
We therefore construct from $G=(\vars,\al,P)$ a context-free grammar $G'=(\vars',\al,P')$ with $\vars\subseteq \vars'$ such that $L(G',X) = \{ w\in\al^\ast \mid \amb_{G,X}(w) = \infty\}$. Then $\{ \vv \in \N^\al\mid \camb_{G,X}(\vv)=\infty \} = \pa(L(G',X))$ and is a semilinear set by Parikh's theorem where the corresponding Presburger formula is again effectively constructible.

We discuss the construction of $G'$ for the sake of completeness: we have $\amb_{G,X}(w)=\infty$ if and only if there are infinitely many $X$-trees $t$ with $\Y(t)=w$.
In particular, for every $h\in\N$ we can find a $X$-tree $t$ of height at least $h$ with $\Y(t)$, as there are only finitely many $X$-trees of bounded height. For instance, choose $h\ge (\abs{w}+1)\abs{\vars}$ and consider a maximal path $v_0\ldots v_h$ from the root of such a $t$ to a leaf. For all $i=0\ldots h$ assume $t|_{v_i}$ is a $X_i$-tree ($X=X_0$). This path then corresponds to a derivation of the form 
\begin{equation*}
X=X_0 \Rightarrow^+ u_0 X_1 v_0 \Rightarrow^+ \ldots \Rightarrow^+ u_0 \ldots u_{h-1} X_h v_{h-1} \ldots v_0 \Rightarrow u_1 \ldots u_{h-1} u_{h}v_{h} v_{h-1} \ldots v_1 = w
\end{equation*}
for suitable $u_i,v_i\in\al^\ast$. In the sequence $X_0,X_1,\ldots,X_h$ color $X_i$ black if $\abs{u_iv_i}=0$; otherwise color $X_i$ red.
Then there are at most $\abs{w}$ red variables in this sequence.
In particular, there is a subsequence $X_{i},X_{i+1},\ldots,X_{i+\abs{\vars}}$ consisting of $1+\abs{\vars}$ consecutive black variables, as otherwise $h+1\le (\abs{w}+1)\abs{\vars}$. Hence, the derivation contains a cyclic derivation $Y\Rightarrow^+ Y$.

Therefore compute the set $\vars_C = \{ X\in\vars \mid X\Rightarrow_G^+ X\}$ of cyclic variables as usual, and define $G'$ such that a derivation can only terminate in a word if the derivation visits at least one cyclic variable:
\begin{itemize}
\item Set $\vars' = \{ X,X' \mid X\in\vars\}$ with the intended meaning that an unprimed variable still has to be derived into a sentential form containing at least one cyclic variable $Y\in\vars_C$.
\item Construct $P'$ as follows:
  \begin{itemize}
  \item
  If $X\to_{G} u_0$ for $u_0\in\al^\ast$, then $X'\to_{G'} u_0$.
  \item
  If $X\to_{G} u_0 X_1 u_1 X_2 u_2 \ldots u_{r-1} X_r u_r$ for $r>0$ and $u_i\in \al^\ast$,
  then
  \begin{equation*}
  X' \to_{G'} u_0 X'_1 u_1 X'_2 u_2 \ldots u_{r-1} X'_r u_r
  \end{equation*}
  and
  \begin{equation*}
  \begin{array}{lcl}
  X &\to_{G'}& u_0 X_1 u_1 X'_2 u_2 \ldots u_{r-1} X'_r u_r\\
  X &\to_{G'}& u_0 X'_1 u_1 X_2 u_2\ldots u_{r-1} X'_r u_r\\
   &\vdots& \\
  X &\to_{G'}& u_0 X'_1 u_1 X'_2 u_2 \ldots u_{r-1} X_r u_r
  \end{array}
  \end{equation*}
  \item
  If $X\in \vars_C$, then $X\to_{G'} X'$.
  \end{itemize}
\end{itemize}
By construction, an unprimed variable $Y$ can only be rewritten to a sentential form containing exactly one unprimed variable, except $Y$ is cylic in $G$, in which case the rule $Y\to_{G'} Y'$ can also be applied.

Then $w\in L(G',X)$ if and only if there is a derivation $X\Rightarrow_{G'}^+ u Y v \Rightarrow_{G'} u Y' v \Rightarrow^+_{G'} w$,
as only primed variables can be rewritten to terminal words. By construction, this is equivalent to $X\Rightarrow_{G}^+ u Y v \Rightarrow_{G}^+ w$ and $Y\in\vars_C$, which in turn is equivalent to $\amb_{G,X}(w)=\infty$.

\end{document}